\documentclass[conference]{IEEEtran}

\usepackage[dvips]{graphicx}
\usepackage{amssymb}
\usepackage{graphicx}
\usepackage[cmex10]{amsmath}
\usepackage{algorithm}
\usepackage{algorithmic}
\usepackage{array}
\usepackage{eqparbox}
\usepackage[tight,footnotesize]{subfigure}
\usepackage{float}
\usepackage{url}
\usepackage{paralist}
\usepackage{multirow}
\usepackage{amssymb}
\usepackage{amsthm}
\newtheorem{thm}{Theorem}
\newtheorem{lem}{Lemma}

\begin{document}

\title{Optimal Auction Design with Quantized Bids}

\author{\IEEEauthorblockN{Nianxia Cao, Swastik Brahma, Pramod K. Varshney\\}
\IEEEauthorblockA{Department of Electrical Engineering and Computer Science, Syracuse University, NY, 13244, USA\\
\{ncao, skbrahma, varshney\}@syr.edu}}

\maketitle

\begin{abstract}
\boldmath This letter considers the design of an auction mechanism
to sell the object of a seller when the buyers quantize their private value estimates regarding the object prior to communicating them to the seller. The designed auction mechanism maximizes the utility of the seller (i.e., the auction is optimal), prevents buyers from communicating falsified quantized bids (i.e., the auction is incentive-compatible), and ensures that buyers will participate in the auction (i.e., the auction is individually-rational). The letter also investigates the design of the optimal quantization thresholds 
using which buyers quantize their private value estimates. 
Numerical results provide insights regarding the influence of the quantization thresholds on the auction mechanism.
\end{abstract}


\section{Introduction}
\label{sec:intro}

The field of mechanism design (also known as reverse game theory) aims to study how to 
implement desired objectives (social or individual) in systems comprised of multiple selfish and rational agents, 
with agents having private information 
that influence the solution \cite{algoMechDesign}. 
Auction design \cite{vijayKrishna}, which falls in the category of mechanism design problems, 
seeks to investigate how to allocate an object (such as, a resource) to a set of buyers, 
with buyers having private value estimates about the object, 
and to determine the price at which to trade the object via competition among the buyers. 
In general, auction design has been a well studied topic in the past. 
A good overview of the topic is provided in \cite{vijayKrishna}.

However, in traditional auction design (such as, in \cite{vijayKrishna}), it has been assumed that the 
buyers send their private information, 
typically considered as analog values, to the seller in analog form.
In contrast, in this letter, we consider the design of an auction mechanism where
the buyers quantize their private information, i.e., their private analog values, 
prior to communication.
Quantization of analog private information prior to communication is practical,
for example, when the buyers and the seller communicate in a resource constrained environment (such as, with limited bandwidth and energy).
Some example scenarios of such 
environments 
include auction based resource allocation for 
sensor management \cite{ncao_Globalsip13,ncao_tsp15,Pricetheorey_Chavali}, 
spectrum allocation in Cognitive Radio systems \cite{kasbekar2010spectrum,kash2014enabling,
nadendla2012auction,huang2014truthful}, 
and routing games in networks \cite{su2009auction,hershberger2001vickrey,
shu2010truthful}. 
It should be noted that, design of an auction mechanism
with quantized bids is not only complicated by the fact that
the seller is unaware of the true value estimates of the bidders,
but also by the fact that the seller only gets quantized data
from the buyers that convey information about their private
value estimates. The buyers (being selfish and rational entities) may
intentionally falsify the quantized bids they transmit in order
to acquire an additional advantage, which further complicated the problem. 
Moreover, as can be expected, choice of the quantization thresholds influences 
the outcome of the auction, so that the design of the optimal quantization thresholds becomes important. 

In this letter, we design an optimal auction mechanism where the buyers quantize 
their analog private value estimates regarding the traded object into binary values prior to communication.
To the best of our knowledge, this is the first work till date to investigate this problem.
Our auction mechanism is comprised of three components-
\begin{inparaenum}[\itshape a\upshape)]
	\item \textit{winner determination function}, which determines the bidder who wins the object,
	\item \textit{payment function}, which determines the payment to be made by each bidder, and,
	\item \textit{quantization thresholds}, which determine how the buyers will quantize their private analog value estimates.
\end{inparaenum}
This letter 
designs the aforementioned components of the auction so that the auction
is optimal (i.e., maximizes the seller's utility),
individually-rational (i.e., rationalizes buyer participation), and incentive-compatible
(i.e., prevents buyers from communicating falsified binary bids).
We also study the influence of the quantization thresholds on the optimal mechanism via simulations.

The rest of this letter is organized as follows. 
In Section II, we decribe the auction model considered. Section III formulates the optimal auction design problem. Section IV analyzes the problem for given quantization thresholds. 
Section V investigates the design of the optimal quantization thresholds.
Finally, Section VI concludes this letter.


\section{Auction Model}
The seller, as an auctioneer in our model, has an object to sell to one of $N$ potential buyers. The buyers, on the other hand, compete to buy the object from the seller, and comprise the set of bidders. Each buyer has a private (analog) value estimate regarding the object, which is unknown to the seller.  
The auction with quantized bids is conducted using the following steps: 
\begin{inparaenum}[\itshape a\upshape)]
\item The seller designs an optimal auction mechanism and the corresponding optimal quantization thresholds; 
\item According to the rules set by the seller, the buyers transmit their quantized bids to the seller; 
\item The seller decides the winner of the auction and how much to charge for the object.
\end{inparaenum}

In this letter, we consider the case where the buyers  quantize their private value estimates and transmit binary bids to the seller to compete for the object. For each buyer $i$, there is some private (analog) value estimate $v_i$ for the object, and the corresponding quantized value estimate is denoted as $\omega_i$. 
We assume that the 
value estimate of buyer $i$ can be described by a probability density function $\tilde{f}_i: [a_i,b_i] \rightarrow \mathbf{R}_+$, where $a_i$ is buyer $i$'s lowest value estimate for the object, and $b_i$ is his highest value estimate, and $-\infty \leq a_i \leq b_i \leq \infty$. The binary value estimate of buyer $i$ is defined as:
\begin{equation}\label{eq:binary_t}
\omega_i = \left\{ \begin{array}{rl}
  0 &\mbox{ $a_i \leq v_i \leq \eta_i$} \\
  1 &\mbox{ $\eta_i < v_i \leq b_i$ }  \\
       \end{array} \right.
\end{equation}
where, $\eta_i$ is the quantization threshold of buyer $i$. The seller's uncertainty about the binary value estimate of buyer $i$ can be described by the probability mass function (pmf) of $\omega_i$
\begin{equation}
f_i(0) = \Pr(\omega_i=0) = F(a_i\leq v_i\leq \eta_i)=\int_{a_i}^{\eta_i} \tilde{f}_i(v_i) d v_i.
\end{equation}
Let $\Omega$ denote the set of all possible combinations of buyers' binary value estimates $\Omega = \{0,1 \}^N$, i.e., the vector $\boldsymbol{\omega} \in \Omega$. Similarly, we let $\Omega_{-i}$ denote the set of all possible combinations of the value estimates of the buyers other than $i$, so that the vector $\boldsymbol{\omega}_{-i} = [\omega_1,\ldots,\omega_{i-1},\omega_{i+1},\ldots,\omega_N]^T \in \Omega_{-i}$ where $\Omega_{-i} = \{ 0,1\}^{N-1}$.

The binary value estimates of the buyers are assumed to be statistically independent with each other. Thus, the joint pmf of the vector $\boldsymbol{\omega}$ is $f(\boldsymbol{\omega}) = \prod_{j=1,\ldots,N} f_j(\omega_j)$. 
We assume that buyer $i$ treats the other buyers' binary value estimates in a similar way as the seller does. Thus, both the seller and buyer $i$ consider the joint pmf of the vector of value estimates for all the buyers other than $i$, $\boldsymbol{\omega}_{-i}$, to be $f_{-i}(\boldsymbol{\omega}_{-i}) = \prod_{j=1,\ldots,i-1,i+1,\ldots,N}f_j(\omega_j)$. 
The seller's personal value estimate for the object is denoted by $v_0$. 

\section{Auction Design Problem Formulation}
The auction design problem is to design the outcome functions $\mathbf{q}$, $\mathbf{p}$, and the quantization thresholds $\boldsymbol{\eta}$ that maximize the seller's expected utility subject to certain constraints, 
where $\mathbf{q},\mathbf{p}:\Omega\to \mathbb{R}^N$, $\mathbf{q} = [q_1, \ldots, q_N]^T$, $\mathbf{p} = [p_1,\ldots,p_N]^T$. Specifically, $q_i(\boldsymbol{\omega})$ is the probability of buyer $i$ being selected by the seller, and $p_i(\boldsymbol{\omega})$ is the amount that buyer $i$ has to pay\footnote{Notice that, in our formulation, we allow for the possibility that a buyer may have to pay something even if he is not selected as the winner, but as we will later show this will not be the case.}. 
Further, in this letter, we focus on the direct mechanism, where the buyers directly transmit their binary bids to the seller \cite{myerson1983efficient}.


By assuming throughout this letter that the seller and the buyers are risk neutral, we next define the utility functions of the seller and the buyers. The expected utility of the seller is, 
\begin{equation}\label{u_0}
\mathcal{U}_0 (\mathbf{p}, \mathbf{q}, \boldsymbol{\eta}) = \sum_{\boldsymbol{\omega}\in \Omega} \left[ v_{0} \left(1-\sum\limits_{i=1}^{N} q_i(\boldsymbol{\omega})\right) + \sum\limits_{i=1}^{N}  p_i(\boldsymbol{\omega})\right] f(\boldsymbol{\omega})
\end{equation}
Since buyer $i$ is aware of his actual value estimate $v_i\in [a_i,b_i]$, his expected utility with the binary bid $\omega_i\in\{0,1\}$ is described as
\begin{equation}\label{ui_t}
\begin{aligned}
&\mathcal{U}_i (p_i, q_i, v_i, \omega_i, \boldsymbol{\eta}_{-i}) = \sum_{\boldsymbol{\omega}_{-i}\in \Omega_{-i}} \Big[ v_i q_i(\boldsymbol{\omega}) - p_i(\boldsymbol{\omega}) \Big] f_{-i}(\boldsymbol{\omega}_{-i})
\end{aligned}
\end{equation}
where $\boldsymbol{\eta}_{-i} = [\eta_1,\ldots,\eta_{i-1},\eta_{i+1},\ldots,\eta_N]^T$. Consider now that buyer $i$'s actual value estimate $v_i$ was supposed to be quantized to $\omega_i$ according to \eqref{eq:binary_t}, but he instead transmits a binary value estimate $\tilde{\omega}_i$ ($\tilde{\omega}_i \in \{0,1\}$, $\tilde{\omega}_i$ needs not to be equal to $\omega_i$). Then, his expected utility would be
\begin{equation}\label{ui_f}
\tilde{\mathcal{U}}_{i} = \sum_{\boldsymbol{\omega}_{-i}\in \Omega_{-i}} \Big[ v_i q_i(\tilde{\omega}_i,\boldsymbol{\omega}_{-i})  - p_i(\tilde{\omega}_i,\boldsymbol{\omega}_{-i}) \Big] f_{-i}(\boldsymbol{\omega}_{-i})
\end{equation}
The optimal auction mechanism is designed to maximize the seller's expected utility while ensuring some conditions: 
\begin{subequations}\label{eq_opt}
\begin{align}
& \underset{\mathbf{p}, \mathbf{q}, \boldsymbol{\eta}}{\text{maximize}} & & \mathcal{U}_0 (\mathbf{p}, \mathbf{q}, \boldsymbol{\eta}) \nonumber\\
& \text{subject to} & & \mathcal{U}_i (p_i, q_i, v_i, \omega_i, \boldsymbol{\eta}_{-i}) \geq 0\label{eq:opt_a}\\
& & & \mathcal{U}_i  \geq \tilde{\mathcal{U}}_{i} \label{eq:opt_b}\\
& & & \sum_{i=1}^N q_i(\boldsymbol{\omega}) \leq 1, \quad \forall \boldsymbol{\omega}\in \Omega \label{eq:opt_d}\\
& & & 0 \leq q_{i}(\boldsymbol{\omega}) \leq 1, \quad i \in \left\{1, \ldots N \right\}\quad \forall \boldsymbol{\omega}\in \Omega\label{eq:opt_e}\\
& & & a_i \leq \eta_{i} \leq b_i, \quad i \in \left\{1, \ldots N \right\}\label{eq:opt_f}
\end{align}
\end{subequations}
where $v_i \in [a_i, b_i]$ and $\tilde{\omega}_i, \omega_i\in\{0,1\}$.  Below we describe each constraint in detail.
\begin{itemize}
\item \textit{Individual Rationality (IR) constraint}~\eqref{eq:opt_a}: We assume that the seller cannot force a buyer to participate in an auction. If the buyer does not participate in the auction, he would not get the object, but also would not pay the seller, so his utility would be zero. Thus, to make buyers participate in the auction, this condition must be satisfied.
\item \textit{Incentive-Compatibility (IC) constraint}~\eqref{eq:opt_b}: We assume that the seller can not prevent any buyer from lying about his binary value estimate if the buyer is expected to gain from lying.
Thus, to make sure that no buyer has any incentive to lie about his value estimate,
transmission of true binary value estimates must form a Bayesian Nash equilibrium of the game.
\item \textit{Probability constraints}~\eqref{eq:opt_d} and \eqref{eq:opt_e}: Since there is only one object, the seller can select at most one buyer to sell his object.
\item \textit{Threshold constraint}~\eqref{eq:opt_f}: The quantization thresholds are between the lowest and highest value estimates of each buyer.
\end{itemize}

\section{Optimal Auction Design for Given Quantization Thresholds}
\label{sec:auction}

In this section, we analyze the optimal mechanism design problem when the quantization thresholds $\boldsymbol{\eta}$ are given. We first state a lemma corresponding to the IC condition of \eqref{eq:opt_b}.
\begin{lem}\label{lem}
The IC condition holds if and only if the following conditions hold:
\begin{subequations}
\begin{align}
&1~~Q_i^1 - Q_i^0 \geq 0 \label{ic2} \\
&2~~P_i^1 = P_i^0 + \eta_i(Q_i^1 - Q_i^0).\label{ic1}
\end{align}
\end{subequations}
where 
$Q_i^l \triangleq \sum_{\boldsymbol{\omega}_{-i}} q_i(\omega_i=l,t_{-i}) f_{-i}(\boldsymbol{\omega}_{-i})$ is the expected probability that buyer $i$ will be selected when he transmits his binary value estimate $l$ conditioned on all other buyers' binary value estimates. Similarly, $P_i^l \triangleq \sum_{\boldsymbol{\omega}_{-i}} p_i(\omega_i=l,t_{-i}) f_{-i}(\boldsymbol{\omega}_{-i})$ is the expected payment buyer $i$ has to pay when he transmits his binary value estimate $l$ conditioned on all other buyers' binary value estimates. 
\end{lem}
\begin{proof}
See Appendix \ref{lem_proof}.
\end{proof}
Note that, the IC conditions in \eqref{ic2} and \eqref{ic1} imply the following: 
\begin{inparaenum}
\item for buyer $i$, the winning probability for transmitting 1 is no less than that for transmitting 0,
\item if $i$ wins, the expected amount he has to pay by transmitting 1 is larger than that of transmitting 0. 
\end{inparaenum} Thus, the IC condition can be understood in the following way: if buyer $i$'s actual value estimate is supposed to be quantized to 0 according to his quantization threshold, then he does not have an incentive to transmit 1 and pay more for the object; if buyer $i$ is supposed to quantize his actual value estimate to 1, he may have an incentive to transmit 0 and pay less (higher utility), however, transmitting 0 instead of 1 will decrease his probability to win the auction.

Based on Lemma 1, we can simplify the auction design problem in \eqref{eq_opt}, when the quantization thresholds are given, as follows. 
\begin{thm}\label{theorem}
The optimal mechanism design problem of \eqref{eq_opt}, when the quantization thresholds are given, is equivalent to
\begin{subequations}\label{eq_opt_2}
\begin{align}
& \underset{\mathbf{q}}{\text{maximize}} & &  \sum_{\boldsymbol{\omega}\in \Omega}\left[ \sum\limits_{i=1}^{N}u_{i}(\omega_i)q_i(\boldsymbol{\omega})\right] f(\boldsymbol{\omega})\label{eq_opt_2_a}\\ 
& \text{subject to} & & \sum_{i=1}^N q_i(\boldsymbol{\omega}) \leq 1, \quad \forall \boldsymbol{\omega}\in \Omega \tag{\ref{eq:opt_d}}\\
& & & 0 \leq q_{i}(\boldsymbol{\omega}) \leq 1, \quad i \in \left\{1, \ldots N \right\}\quad \forall \boldsymbol{\omega} \in \Omega \tag{\ref{eq:opt_e}}
\end{align}
\end{subequations}
where 
\begin{equation}\label{eq_u}
u_{i}(\omega_i) = \left\{ 
  \begin{array}{l l}
     \dfrac{a_i - (1-\lambda_i)\eta_i}{\lambda_i} - v_{0}  , & \text{ $\omega_i=0$}\\
    \eta_i - v_{0}, & \text{$\omega_i=1$}
  \end{array} \right.
\end{equation}
with $\lambda_i = f_i(\omega_i=0)$, and the payment to buyer $i$ is given by
\begin{equation}\label{eq:pi}
p_i(\omega_i,\boldsymbol{\omega}_{-i}) = \eta_i q_i(\omega_i,\boldsymbol{\omega}_{-i}) - (\eta_i - a_i)q_i(\omega_i=0,\boldsymbol{\omega}_{-i})
\end{equation}
\end{thm}
\begin{proof}
See Appendix \ref{Thm_proof}.
\end{proof}
Based on Theorem \ref{theorem}, when the quantization thresholds are given, the optimal auction mechanism can be described as follows:
\begin{itemize}
\item For any set of realizations of the binary value estimates $\boldsymbol{\omega}$, the seller compares the corresponding $u_i$s (defined based on \eqref{eq_u}), and sells the object to the buyer with the highest $u_i$. In other words, if $u_i(\omega_i)$ is the highest among all the buyers, then the solution of the winning probability $\mathbf{q}$ is: $q_i=1$, $q_j=0$ for $\forall j\in \{1,\cdots,i-1, i+1,\cdots, N\}$.
\item Only the buyer that wins the auction needs to pay the seller for the object, and the payment is,
\begin{equation}
\begin{aligned}
p_i &= a_i \quad \text{if buyer } i \text{ wins by transmitting 0}\\
p_i &= \eta_i - (\eta_i - a_i)q_i(\omega_i=0,\boldsymbol{\omega}_{-i})\\
& \qquad \qquad \text{if buyer } i \text{ wins by transmitting 1}
\end{aligned}
\end{equation}
Note that, if buyer $i$ wins the auction by transmitting 1, the seller needs to further determine $q_i(\omega_i=0,\boldsymbol{\omega}_{-i})$ to compute the payment, i.e., determine if buyer $i$ would have still won the auction had his binary bid been 0 for the same set of binary bids of the other bidders. 
\item If there is a tie, i.e., multiple bidders have the highest $u_i$, the seller can arbitrarily choose a winner among them without affecting his own utility.
\end{itemize} 




\section{Optimal Quantization Thresholds}
\label{sec_OptEta}

In Section \ref{sec:auction}, we have designed an optimal mechanism when the quantization thresholds $\boldsymbol{\eta} = [\eta_1,\ldots,\eta_N]^T$ are given. From Theorem \ref{theorem}, we observe that the thresholds influence the outcome of the auction mechanism. In this section, we investigate the design of the optimal quantization thresholds by assuming that the value estimates of the buyers are uniformly distributed. We also study the impact of the quantization thresholds by applying our auction model for spectrum auctions in Section V-A.
 

We first study the case when there is only one buyer who is interested in the object. The value estimate of the buyer is assumed to be in $[a,b]$.
With only one buyer, the objective function \eqref{eq_opt_2_a} in the optimization problem \eqref{eq_opt_2} is (note that the indices have been omitted)
\begin{equation}\label{eq:obj_onebuyer}
\begin{aligned}
&u(\omega=0) q(0) \lambda + u(\omega=1) q(1) (1-\lambda)\\
&= \dfrac{(\eta-a)(\eta-b-v_0+a)}{b-a} q(0) + \dfrac{(b-\eta)(\eta - v_0)}{b-a} q(1)\\
\end{aligned} 
\end{equation}
The seller maximizes \eqref{eq:obj_onebuyer} over $\eta$ and $(q(0), q(1))$ subject to the constraints that $a \leq \eta \leq b$, $0 \leq q(0) \leq 1$, $0 \leq q(1) \leq 1$. The following lemma gives the optimal quantization threshold $\eta_{opt}$ when only one buyer is interested in the object with different parameter settings.
\begin{lem}\label{lemOneBuyer}
The seller keeps the object instead of selling it to the buyer if $v_0>b$. When $v_0 \leq b$, the seller designs the optimal quantization threshold $\eta$: if $(b+v_0)/2 > a$, the optimal value for $\eta$ is $(b+v_0)/2$, otherwise, $\eta$ can be any value in $[a,b]$. The details are shown in Table \ref{tab_solution_2}.
\end{lem}
\begin{proof}
See Appendix \ref{OneBuyer_proof}.
\end{proof}
\begin{table}[tb]\caption{Optimal thresholds with one buyer}
\label{tab_solution_2}
\centering
\begin{tabular}{ | c | c | c | c |}
    \hline
    $v_0, a, b$ & $\eta_{opt}$ & $\mathcal{U}_0^{opt}$ \\ \hline
    $(b+v_0)/2 > a$  & $(b+v_0)/2$ & $(b-v_0)^2/(4(b-a))$  \\  \hline
    $(b+v_0)/2 \leq a$ & $\forall$ & $a-v_0$\\ \hline
\end{tabular}
\end{table}



\begin{figure}[tb]
\begin{center}
\includegraphics[
  width=0.7\columnwidth, height = !]{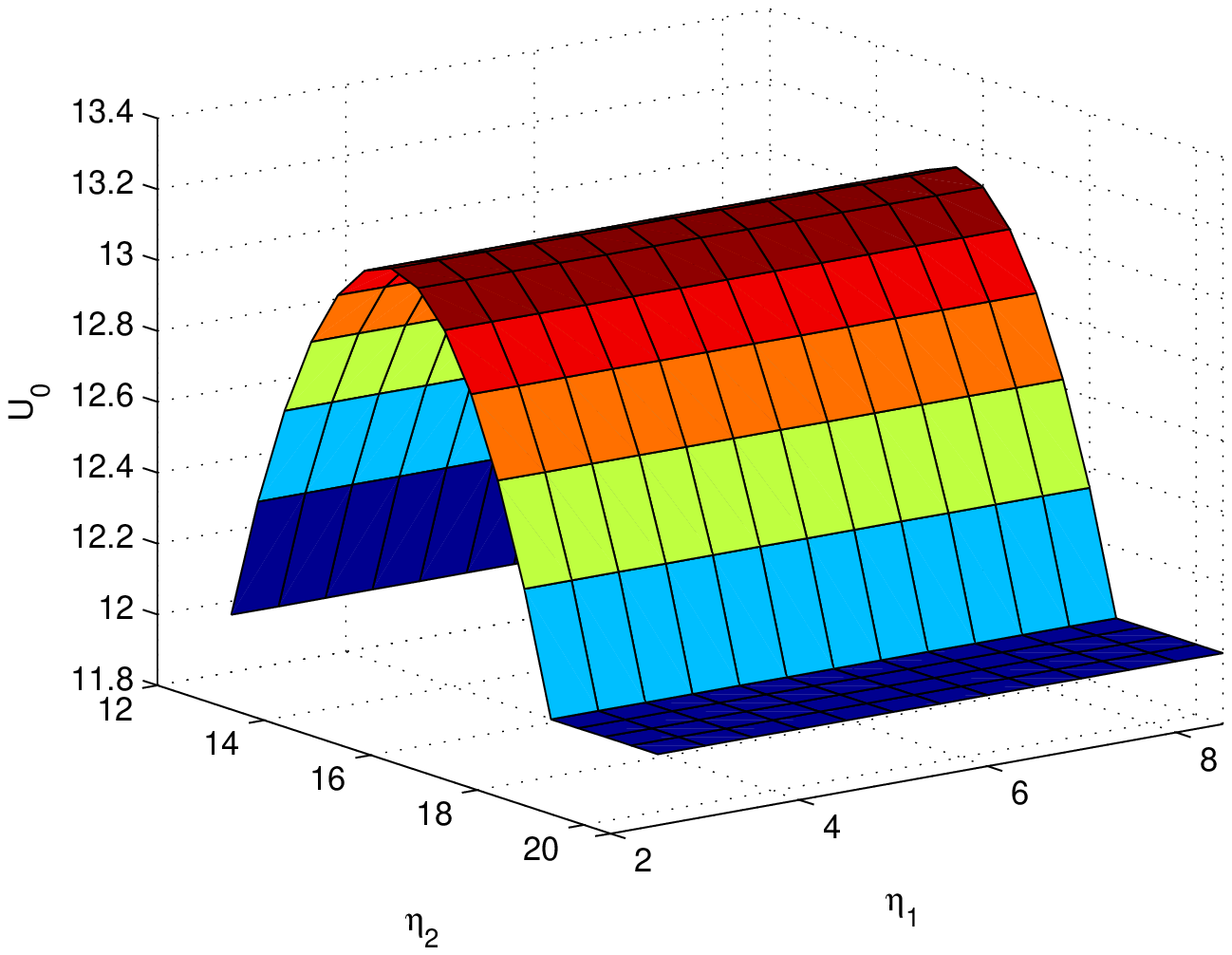}
\caption{Utility of the seller as function of the nonidentical thresholds, $N=2$.}
\label{fig_nonidentical2}
\end{center}
\end{figure}

\begin{figure}[tb]
\begin{center}
\includegraphics[
  width=0.7\columnwidth, height = !]{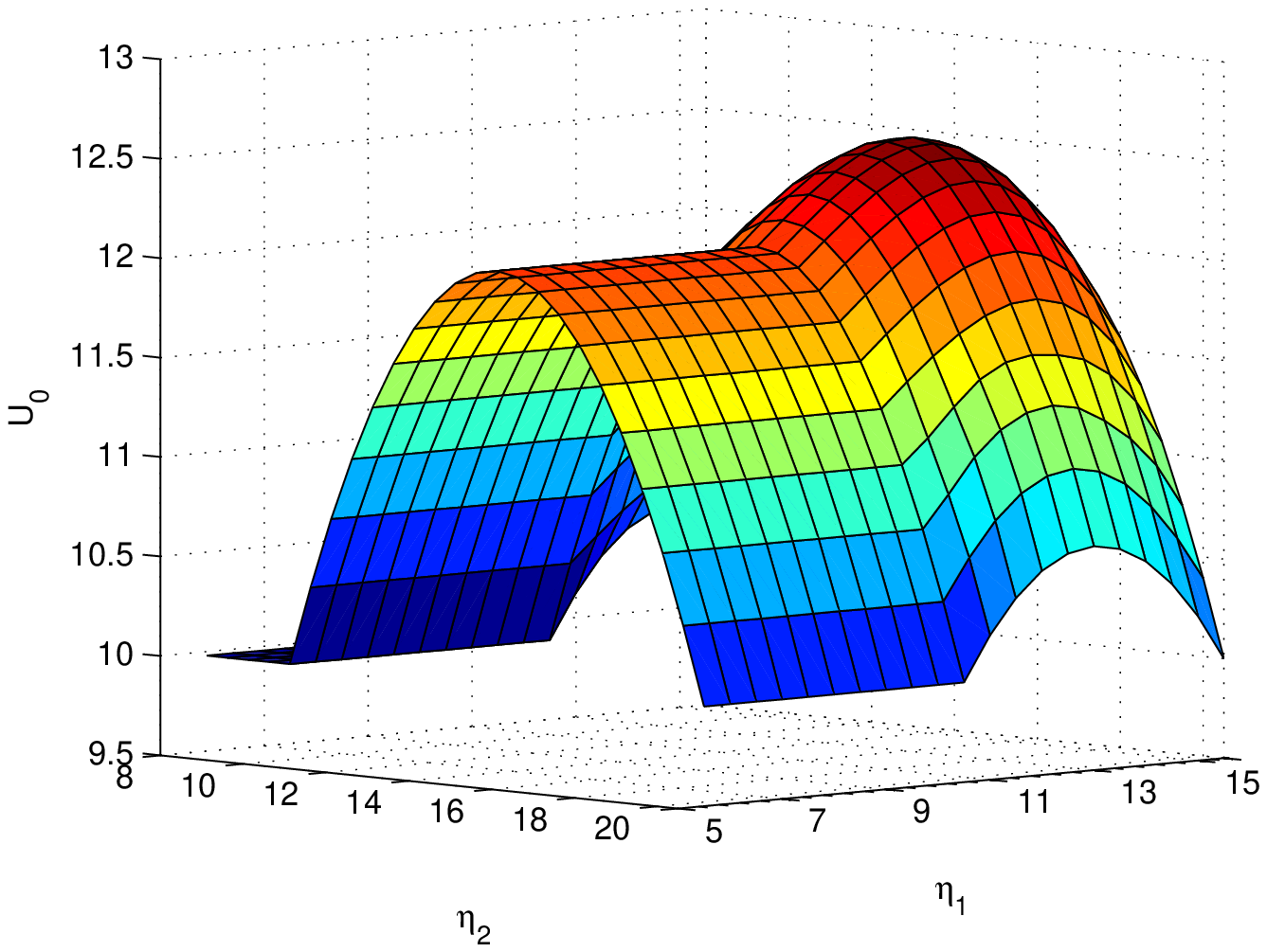}
\caption{Utility of the seller as function of the nonidentical thresholds, $N=2$.}
\label{fig_nonidentical}
\end{center}
\end{figure}

Next we consider the scenario when there are multiple bidders. The seller optimizes the expected value in \eqref{eq_opt_2_a}, where the problem of finding the optimal $\mathbf{q}$ when $\boldsymbol{\eta}$ is given in \eqref{eq_opt_2} is a linear optimization problem, which can be solved with the MATLAB function ``linprog''. To illustrate the influence of the quantization thresholds on the auction mechanism, we next provide some numerical examples. First, we assume that there are $N=2$ buyers, $v_0=10$, $a_1=2, b_1=8$, and $a_2=12, b_2=20$. The expected utility as a function of $\boldsymbol{\eta}=[\eta_1,\eta_2]^T$ is shown in Fig. \ref{fig_nonidentical2}.  Since $v_0 > b_1$ and $b_1 < a_2$, the seller would always select buyer 2 as the winning bidder. In this case, it is irrelevant for the seller to consider buyer 1's actual (binary) value estimate, and thus any $\eta_1 \in [a_1,b_1]$ is equally good for the seller. Therefore, as can be seen from Fig. \ref{fig_nonidentical2}, the expected utility of the seller changes only with buyer 2's quantization threshold, and is invariant of buyer 1's quantization threshold. From Fig. \ref{fig_nonidentical2}, it can also be seen that the optimal threshold for buyer 2 is $\eta_2=15$.

In Fig. \ref{fig_nonidentical}, we study the scenario where $v_0=10$, $a_1=5, b_1=15$, and $a_2=8, b_2=20$. The expected utility of the seller is a function of both buyer's  quantization thresholds, since the interval of the two buyers' value estimates are overlapped. As can be seen from the figure, the seller sets the quantization thresholds to be $\eta_1=13, \eta_2=15$ to obtain the optimal expected utility.

\subsection{Spectrum Auctions: Impact of Bid Quantization}
We now study the impact of bid quantization by applying our auction model to spectrum auction~\cite{zhang2013auction}
in cognitive radio systems, where the traded object is spectrum.
In the spectrum auction problem, a primary user (PU) when not using a certain portion 
of his assigned spectrum, solicits bids from a set of secondary users (SU) to 
sell the 
unused spectrum (refer to \cite{zhang2013auction} for an overview of spectrum auctions). 
As was mentioned earlier, past literature~\cite{zhang2013auction} on the design of spectrum auctions 
only considers communication of 
analog bids from SUs and has not considered the 
design of optimal
auctions when SUs transmit quantized bids.
To gain insights into the dynamics of auction mechanism with quantized bids in the context of spectrum auctions,
we study the expected utility of the PU (i.e., the seller) when the SUs (i.e., the bidders)
send analog bids as well as when
the SUs transmit quantized bids.

Fig. \ref{fig:utility} presents the expected utility of the PU under both scenarios. 
The bid of each bidder is assumed to be uniformly distributed over $[5,20]$. We base our results, for the case when SUs send analog bids (`Analog' in Fig. \ref{fig:utility}), on the spectrum auction proposed in~\cite{optimal_auction_myerson}. For the case when SUs send quantized bids, we consider optimal quantization thresholds (`Binary Optimal' in Fig. \ref{fig:utility}) as well as quantization thresholds chosen randomly in an uniform manner (`Binary Random' in Fig. \ref{fig:utility}).   
We observe that for all cases, the expected utility of the PU increases as 
the number of bidders increases. This is because with increasing number of  SUs, the chances of the PU finding a bidder with higher value estimate increase. As expected, the expected utility of the PU with quantized bids is lower than that with analog bids. Also, the use of optimal quantization thresholds result in a higher utility at the seller than when thresholds are chosen randomly.

\begin{figure}[tb]
\begin{center}
\includegraphics[
  width=0.7\columnwidth, height = !]{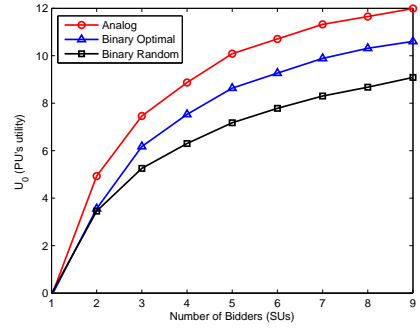}
\caption{Utility of the seller (PU) with number of bidders (SUs).}
\label{fig:utility}
\end{center}
\end{figure}

\section{Conclusion}
In this letter, we designed an optimal auction mechanism for an environment where bidders quantize their value estimates regarding the traded object into binary values prior to communicating them to the seller. The mechanism is designed to maximize the seller's expected utility while ensuring the individual rationality (IR) and incentive-compatibility (IC) constraints. The letter also investigated the design of the optimal quantization thresholds, using which buyers would quantize their private value estimates, such that the seller's expected utility is maximized.

\bibliography{journal}
\bibliographystyle{IEEEtran}

\clearpage

\appendices
\section{Proof of Lemma \ref{lem}}
\label{lem_proof}
Recall \eqref{eq:opt_b} and \eqref{eq:pi}, we get the equivalent IC conditions as
\begin{subequations}
\begin{align}
&v_i Q_i^0 - P_i^0 \geq v_i Q_i^1 - P_i^1, \quad \forall v_i \in [a_i, \eta_i] \label{icp_1}\\
&w_i Q_i^1 - P_i^1 \geq w_i Q_i^0 - P_i^0, \quad \forall w_i \in (\eta_i, b_i]\label{icp_2}
\end{align}
\end{subequations}
Since \eqref{icp_1} and \eqref{icp_2} can be directly written as
\begin{subequations}
\begin{align}
&v_i (Q_i^1 - Q_i^0) \leq  P_i^1 - P_i^0, \quad \forall v_i \in [a_i, \eta_i] \label{icp_3} \\
&w_i (Q_i^1 - Q_i^0) \geq  P_i^1 - P_i^0, \quad \forall w_i \in (\eta_i, b_i] \label{icp_4}
\end{align}
\end{subequations}
If $Q_i^0 = Q_i^1$, \eqref{icp_3} and \eqref{icp_4} imply the condition that $P_i^0 = P_i^1$. If $Q_i^1 < Q_i^0$, \eqref{icp_3} and \eqref{icp_4} are equivalent to
\begin{align*}
&v_i \geq \dfrac{P_i^1 - P_i^0}{Q_i^1 - Q_i^0}, \quad \forall v_i \in (a_i, \eta_i] \\
&w_i \leq \dfrac{P_i^1 - P_i^0}{Q_i^1 - Q_i^0}, \quad \forall w_i \in (\eta_i, b_i)
\end{align*}
Then the condition
$b_i \leq \dfrac{P_i^1 - P_i^0}{Q_i^1 - Q_i^0} \leq a_i$
must be satisfied, which is contradictory to our definition of buyers' value estimates. 
With $Q_i^1 > Q_i^0$, we have
\begin{equation}\label{icp_5}
\begin{aligned}
&v_i \leq \dfrac{P_i^1 - P_i^0}{Q_i^1 - Q_i^0},\quad \forall v_i \in [a_i, \eta_i] \\
&w_i \geq \dfrac{P_i^1 - P_i^0}{Q_i^1 - Q_i^0},\quad \forall w_i \in (\eta_i, b_i].
\end{aligned}
\end{equation}
So that $Q_i^1\geq Q_i^0$, and
$\eta_i = \dfrac{P_i^1 - P_i^0}{Q_i^1 - Q_i^0}$.
Thus the lemma is proved.

\section{Proof of Theorem \ref{theorem}}
\label{Thm_proof}
The IR constraint of \eqref{eq:opt_a} can be considered for the two cases as: 
\begin{subequations}
\begin{align}
&\mathcal{U}_i (p_i, q_i, v_i, \omega_i = 0) \geq 0,\quad \forall v_i \in [a_i, \eta_i] \label{ui_0}\\
&\mathcal{U}_i (p_i, q_i, v_i, \omega_i = 1)  \geq 0,\quad \forall v_i \in [\eta_i, b_i]. \label{ui_1}
\end{align}
\end{subequations}
We may write the seller's objective function of \eqref{eq_opt} as
\begin{equation}\label{thmp1}
\begin{aligned}
&\mathcal{U}_0 (\mathbf{p}, \mathbf{q}) \\
&= v_0 - \sum\limits_{i=1}^{N}  v_0 \left(\sum_{\boldsymbol{\omega}\in \Omega}q_i(\boldsymbol{\omega})f(\boldsymbol{\omega}) \right)+ \sum\limits_{i=1}^{N} \sum_{\boldsymbol{\omega}\in \Omega} p_i(\boldsymbol{\omega}) f(\boldsymbol{\omega}) \\
&= v_0 -\sum\limits_{i=1}^{N} v_0 \Big[\lambda_i Q_i^0  + (1-\lambda_i) Q_i^1 \Big]\\
&\qquad + \sum\limits_{i=1}^{N} \left[\lambda_i P_i^0 + (1-\lambda_i)P_i^1 \right]
\end{aligned}
\end{equation}
By \eqref{ic1} in Lemma 1, we know that
\begin{equation}\label{thmp2}
\begin{aligned}
&\lambda_i P_i^0 + (1-\lambda_i)P_i^1 \\
=& \lambda_i P_i^0 + (1-\lambda_i)\Big[P_i^0 + \eta_i (Q_i^1 - Q_i^0)\Big]\\
=& P_i^0 + (1-\lambda_i) \eta_i (Q_i^1 - Q_i^0)
\end{aligned}
\end{equation}
The expected payment of buyer $i$ for $\forall v_i \in [a_i,\eta_i]$ is
\begin{equation*}
P_i^0 = -\mathcal{U}_i (p_i, q_i, v_i, \omega_i = 0) + v_i Q_i^0,
\end{equation*}
with $v_i=a_i$, 
\begin{equation}\label{pi1}
P_i^0=-\mathcal{U}_i (p_i, q_i, v_i=a_i, \omega_i = 0) + a_i Q_i^0
\end{equation}
Substituting \eqref{thmp2} and \eqref{pi1} into \eqref{thmp1} gives us:
\begin{equation}\label{thmp3}
\begin{aligned}
&\mathcal{U}_0 (\mathbf{p}, \mathbf{q}) \\
=& - \sum\limits_{i=1}^{N} v_0 \Big[\lambda_i Q_i^0 + (1-\lambda_i) Q_i^1\Big]\\
& + \sum\limits_{i=1}^{N} \Big[ (1-\lambda_i) \eta_i (Q_i^1 - Q_i^0) + a_i Q_i^0 \Big] \\
& + v_0 - \sum\limits_{i=1}^{N}\mathcal{U}_i (p_i, q_i, v_i=a_i, \omega_i = 0) \\
=&\sum\limits_{i=1}^{N} \Bigg\{\lambda_i \Big[\dfrac{a_i-(1-\lambda_i)\eta_i}{\lambda_i}-v_0\Big] Q_i^0 + (1-\lambda_i) (\eta_i-v_0)Q_i^1\Bigg\}\\
& + v_0 - \sum\limits_{i=1}^{N}\mathcal{U}_i (p_i, q_i, v_i=a_i, \omega_i = 0)
\end{aligned}
\end{equation}
In \eqref{thmp3}, the payment vector only appears in the last term of the utility of the seller. Also, by the IR constraint \eqref{ui_0}, we know that
\begin{equation}\label{thm_ir}
\mathcal{U}_i (p_i, q_i, v_i=a_i, \omega_i = 0) \geq 0, \quad i \in \{1,\ldots,N\}
\end{equation}
Therefore, to maximize \eqref{thmp3} subject to the constraints, the winning buyer must make payment to the seller according to:
\begin{equation} \label{uie0}
\mathcal{U}_i (p_i, q_i, v_i=a_i, \omega_i = 0) = 0
\end{equation}
which, combined with \eqref{pi1} and \eqref{ic1}, gives the following payment functions 
\begin{equation}\label{pay_2}
\begin{aligned}
&P_i^0 = a_i Q_i^0\\
&P_i^1 = \eta_i Q_i^1 - (\eta_i - a_i)Q_i^0
\end{aligned}
\end{equation}
From \eqref{pay_2}, we get the payment of buyer $i$ regarding his binary value estimate $\omega_i$ in \eqref{eq:pi}.
Further, substituting the payment functions \eqref{pay_2} into the objective function \eqref{thmp1}, we get \eqref{eq_opt_2_a} and \eqref{eq_u}.

To further check the condition in \eqref{ic2}, we observe that
\begin{equation}
\eta_i \geq \dfrac{a_i-(1-\lambda_i) \eta_i}{\lambda_i}
\end{equation}
So that $u_i(\omega_i=1) \geq u_i(\omega_i=0)$, which means that whenever buyer $i$ could win the auction by transmitting a binary value estimate $0$, he could also win if he changed it to $1$. That is, given other buyers' binary value estimates, the expected probability that buyer $i$ would win when he transmits his value estimate to be 1 is higher than that when he transmits 0, i.e., \eqref{ic2} is satisfied. Moreover, \eqref{ic1} is considered in \eqref{thmp2}, and the IR condition is satisfied as shown in \eqref{uie0}. Therefore, the optimization problem considered in \eqref{eq_opt} is equivalent to maximizing the objective function \eqref{thmp3} subject to the buyer selection probability constraints \eqref{eq:opt_d} and \eqref{eq:opt_e}. This proves the theorem.

\begin{table*}[htb]\caption{Optimal quantization thresholds with one buyer.}
\label{tab_solution}
\centering
\begin{tabular}{ | c | c | c | c | c | c |}
    \hline
    $v_0, a, b$ & $\eta$ & $q_{opt}(0)$ & $q_{opt}(1)$ & $\eta_{opt}$ & $\mathcal{U}_0^{opt}$ \\ \hline
    $v_0>b$~($(b+v_0)/2>b$) & $a \leq \eta \leq b$ & 0 & 0 & $\forall$ & 0 \\ \hline
    \multirow{4}{*}{$a<v_0\leq\frac{b+v_0}{2}\leq b < b-a+v_0$} & $a<\eta<v_0$ & 0 & 0 & $\forall$ & 0  \\  
    & $v_0<\eta<b$ & 0 & 1 & $(b+v_0)/2$ & $(b-v_0)^2/(4(b-a))$\\
    & $\eta=b$ OR $\eta=v_0$ & 0 & $\forall$ & $\forall$ & 0\\ 
    & $\eta=a$ & $\forall$ & 0 & $\forall$ & 0\\ \hline
    \multirow{2}{*}{$v_0=a$} & $\eta=a$ OR $\eta=b$ & $\forall$ & $\forall$ & $\forall$ & 0\\ 
    & $a<\eta<b$ & 0 & 1 & $(b+v_0)/2$ & $(b-v_0)^2/(4(b-a))$\\ \hline
    \multirow{4}{*}{$v_0<a <\frac{b+v_0}{2}< b-a+v_0 <b$ } & $\eta=a$ OR $\eta=b-a+v_0$ & $\forall$ & 1 & $a$ OR $b-a+v_0$ & $a-v_0$\\ 
    & $a<\eta<b-a+v_0$ &0 & 1 & $(b+v_0)/2$ & $(b-v_0)^2/(4(b-a))$\\ 
    & $b-a+v_0<\eta<b$ & 1 & 1 & $\forall$ & $a-v_0$\\
    & $\eta=b$ &1 & $\forall$ & $b$ & $a-v_0$\\ \hline
    \multirow{3}{*}{$v_0 < b-a+v_0 \leq \frac{b+v_0}{2} \leq a < b$} & $\eta=a$ & $\forall$ & 1 & $a$ &  $a-v_0$\\ 
    & $\eta=b$ & 1 & $\forall$ & $b$ & $a-v_0$\\
    & $a<\eta<b$ & 1 & 1 & $\forall$ & $a-v_0$\\ \hline
\end{tabular}
\end{table*}
\section{Proof of Lemma \ref{lemOneBuyer}}
\label{OneBuyer_proof}
With only one buyer, the objective function \eqref{eq_opt_2_a} in the optimization problem \eqref{eq_opt_2} is (note that the index has been omitted)
\begin{equation}\label{eq:obj_onebuyer_2}
\begin{aligned}
&u(\omega=0) q(0) \lambda + u(\omega=1) q(1) (1-\lambda)\\
& = \lambda\Big[\dfrac{a - (1-\lambda)\eta}{\lambda} - v_{0}\Big] q(0) + (1-\lambda)(\eta - v_0) q(1)\\
&= \dfrac{\eta^2 -(b+v_0)\eta + a v_0 + a(b-a)}{b-a} q(0) \\
&\qquad + \dfrac{-\eta^2 +(b+v_0)\eta - b v_0}{b-a} q(1)\\
&= \dfrac{(\eta-a)(\eta-b+a-v_0)}{b-a} q(0) + \dfrac{(b-\eta)(\eta - v_0)}{b-a} q(1)\\
&\triangleq M q(0) + N q(1)\\
\end{aligned} 
\end{equation}
where $M \triangleq [(\eta-a)(\eta-b+a-v_0)]/(b-a)$ and $N \triangleq [(b-\eta)(\eta - v_0)]/(b-a)$, 
and $a \leq \eta \leq b$, $0 \leq q(0) \leq 1$, $0 \leq q(1) \leq 1$, 

We observe that the optimal solutions for $q_{opt}(0)$, $q_{opt}(1)$, and $\eta_{opt}$ depend on the relationship among the parameters of the system. Thus, we list all the conditions and the corresponding solutions in Table \ref{tab_solution}\footnote{Notation ``$\forall$'' for $q_{opt}$ represents that $q_{opt}$ can be either 1 or 0, and notation ``$\forall$'' for $\eta_{opt}$ represents that $\eta_{opt}$ can be any value between $[a,b]$.}. From the table, it can be observed that when $v_0 > b$ (row 1 of Table \ref{tab_solution}), the seller does not sell the object, so that design of the quantization thresholds is irrelevant. Otherwise, if $(b+v_0)/2 > a$ (rows 2, 3 and 4), then, since it can be shown that $(b-v_o)^2/4(b-a) \geq a-v_0$ for any real values of $a$, $b$, and $v_0$, the optimal quantization threshold can be set as $(b+v_0)/2$. However, if $(b+v_0)/2 \leq a$ (row 5), any value in [a,b] is equally good as a quantization threshold. This proves the lemma.

%

\end{document}